\documentclass[12pt]{paper}

\usepackage{amsmath}
\usepackage{amsfonts}
\usepackage{latexsym}
\usepackage{amssymb}
\usepackage{amsthm}
\usepackage{pb-diagram}

\usepackage{lineno,xcolor}
\makeatletter

\newtheorem{theorem}{Theorem}[section]
\newtheorem{proposition}[theorem]{Proposition}
\newtheorem{corollary}[theorem]{Corollary}
\newtheorem{lemma}[theorem]{Lemma}

\newtheorem{remark}[theorem]{Remark}

\def\1{\mathbf 1}

\def\a{\mathbf a}
\def\b{\mathbf b}

\def\z{\mathbf z}
\def\0{\mathbf 0}
\def\cA{\mathcal A}
\def\cB{\mathcal B}
\def\cC{\mathcal C}
\def\cD{\mathcal D}

\def\cF{\mathcal F}

\def\cL{\mathcal L}

\def\cN{\mathcal N}
\def\cK{\mathcal K}

\def\cX{\mathcal X}

\def\cS{\mathcal S}

\def\bZ{{\mathbb Z}}

\def\PG{{\rm PG}}
\def\PGL{{\rm PGL}}

\def\GF{{\rm GF}}
\def\GL{{\rm GL}}

\def\SL{{\rm SL}}
\def\PSL{{\rm PSL}}
\def\Sp{{\rm Sp}}

\def\Fq{\mathbb F_{q}}

\def\Fqm{\mathbb F_{q^m}}

\def\Fqthree{\mathbb F_{q^3}}

\def\Aut{{\rm Aut}}

\def\diag{{\rm diag}}
\def\dim{{\rm dim}}

\def\End{{\rm End}}

\def\GF{{\rm GF}}
\def\Im{{\rm Im}}

\def\Rad{{\rm Rad}}
\def\rank{{\rm rk}}

\def\Tr{{\rm Tr}}

\newcommand\comment[1]{}

\def\<{\langle}
\def\>{\rangle}

\title{Non-linear maximum rank distance codes in the cyclic model for the field reduction of finite geometries}

\author{N. Durante and A. Siciliano}

\begin{document}

\maketitle

\begin{abstract}
In this paper we construct  infinite families of non-linear maximum rank distance  codes by using the setting of bilinear forms of a finite vector space.  We also give a geometric description of such codes by using the cyclic model for the field reduction of finite geometries and we show that these families contain the non-linear maximum rank distance  codes recently provided by Cossidente, Marino and Pavese.
\end{abstract}
 


%
%
%
%
%
%
%

\section{Introduction}


Let $M_{m,m'}(\Fq)$,  $m\le m'$, be the rank metric space  of all the $m\times m'$ matrices with entries in the finite field $\Fq$ with $q$ elements, $q=p^h$, $p$ a prime. The {\em distance} between two matrices  by definition  is the rank of their difference. 
An  $(m,m',q;s)${\em -rank distance code} (also {\em rank metric code}) is any subset $\cX$ of $M_{m,m'}(\Fq)$ such that the  minimum  distance between two of its distinct elements is $s+1$.  
An $(m,m',q;s)$-rank distance code is said to be {\em linear} if it is a linear subspace of $M_{m,m'}(\Fq)$. 
 
  It is known \cite{ds} that the size of an  $(m,m',q;s)$-rank distance code  $\cX$ is bounded by  the {\em Singleton-like bound}: 
\[
|\cX| \le q^{m'(m-s)}.
\]
When this bound is achieved, $\cX$ is called  an  $(m, m', q; s)${\em -maximum rank distance code}, or $(m, m', q; s)$-{\em MRD code}, for short.

Although MRD  codes are very interesting by their own and they caught the attention of many researchers  in recent years \cite{alr,ckww,rav}, such codes  have also   applications in error-correction for random network
coding  \cite{gpt,kk,skk}, space-time coding \cite{tsc}  and  cryptography \cite{gpt2,sk}.

Obviously, investigations of MRD codes can be carried out in any rank metric space isomorphic to $M_{m,m'}(\Fq)$.  
In his pioneering paper \cite{ds}, Ph. Delsarte  constructed linear MRD codes for all the possible values of the parameters $m$, $m'$, $q$ and  $s$ by using the framework  of  bilinear forms on two  finite-dimensional vector spaces over a finite  field (Delsarte used the terminology {\em Singleton systems} instead of maximum rank distance codes).

Few years later, Gabidulin   \cite{gab} independently  constructed Delsarte's linear MRD codes as 
evaluation codes of linearized polynomials over a finite field \cite{ln}. That  construction was generalized in \cite{kg} and these codes  are now known as {\em Generalized Gabidulin codes}.

In the case $m'=m$,  a different construction of Delsarte's MRD codes was given by  Cooperstein   \cite{coop}  in the framework of the tensor product of a vector space over $\Fq$ by itself. 
Very recently,   Sheekey   \cite{she} and  Lunardon,  Trombetti and Zhou \cite{ltz} provide some new linear MRD codes by using   linearized polynomials over $\Fqm$.

In finite geometry, $(m,m,q;m-1)$-MRD codes are known as {\em spread sets} \cite{d}.
 To the extent of our knowledge the only
 non-linear MRD codes  that are not spread sets are the $(3,3,q;1)$-MRD codes  constructed   by  Cossidente, Marino and Pavese in \cite{cmp}.  They got such  codes by looking at the geometry of certain algebraic curves of  the projective plane $\PG(2,q^3)$. Such curves, called   {\em  $C_F^1$-sets},  were   introduced and studied by Donati and Durante in \cite{dd}. 
In this paper,  we construct infinite families of non-linear $(m,m,q;m-2)$-MRD codes, for $q\ge 3$ and $m\ge 3$. We also  show that  the Cossidente, Marino and Pavese non-linear MRD  codes belong to these families.
Our  investigation will carry  out in the framework of bilinear forms on a finite dimensional vector space  over $\Fq$.
 
  Let $\Omega=\Omega(V,V)$ be the set of all bilinear forms  on $V$, where $V=V(m,q)$ denotes an $m$-dimensional vector space over $\Fq$. Clearly, $\Omega$ is an $m^2$-dimensional vector space over $\Fq$.
  
The {\em left radical } $\Rad(f)$ of any $f\in\Omega$   by definition is  the subspace of $V $ consisting of all vectors $v$ satisfying $f(v,v' )=0$ for every $v' \in V $. The {\em rank} of $f$ is the codimension of $\Rad(f)$, i.e.
\[
\rank(f)=m-\dim_{\Fq}(\Rad(f)).
\]

Let $u_1,\ldots,u_{m}$   be a basis of $V$. For a given $f\in\Omega$,  the matrix  $(f(u_i,u_j))_{i,j=1,\ldots,m}$, is called the {\em matrix } of $f$ in the basis $u_1,\ldots, u_{m}$ and the map 
\[
\begin{array}{rccc}
\nu=\nu_{\{u_1,\ldots, u_{m}\}}:& \Omega & \rightarrow &  M_{m,m}(\Fq)\\
 & f & \mapsto  & (f(u_i,u_j))_{i,j=1,\ldots,m}
\end{array}
\]
 is an isomorphism of rank metric spaces giving $\rank(f)=\rank(\nu(f))$.

 The group $H=\GL(V)\times \GL(V)$ acts  on $\Omega$ as  a subgroup of $\Aut_{\Fq}(\Omega)$:  for every $(g,g')\in H$, the $(g,g')-$image of any $f\in\Omega$ is defined to be the bilinear form $f^{(g,g')}$ given by 
 \[
 f^{(g,g')}(v,v')=f(gv,g'v').
 \]

Any $\theta\in\Aut(\Fq)$ naturally defines a semilinear transformation of $V$.   For any $f\in\Omega$ and $\theta\in\Aut(\Fq)$, we can define the bilinear form $f^{\theta}(v,v')=f(v^{\theta^{-1}},{v'}^{\theta^{-1}})^\theta$.

The involutorial operator $\top:f\in\Omega\rightarrow f^\top\in\Omega$, where $f^\top$ is given by
 \[
 f^\top(v,v')=f(v',v),
 \]
 is an automorphism of $\Omega$. It turns out that the above automorphisms are all the elements in  $\Aut_{\Fq}(\Omega)$, i.e.  $\Aut_{\Fq}(\Omega)=(\GL(V)\times\GL(V))\rtimes \<\top\>\rtimes\Aut(\Fq)$. 

Two MRD codes $\cX_1$ and $\cX_2$ are said to be {\em equivalent} if there exists $\varphi\in
\Aut_{\Fq}(\Omega)$ such that $\cX_2=\cX_1^\varphi$.

This paper is organized as follows.
In Section \ref{sec_3} we introduce a cyclic model of $\Omega$. In this model we construct infinite families of non-linear MRD codes.
More precisely, for   $q\ge3$,   $m\ge 3$ and $I$  any subset of $\Fq\setminus \{0,1\}$, we provide a subset $\cF_{m,q;I}$ of $\Omega$ which turns out to be  a non-linear $(m,m,q;m-2)$-MRD code (Theorem \ref{th_1}).

In Section \ref{sec_4} we give a geometric description of such codes. If a given rank distance code $\cX$ is considered as a subset of $V(m^2,q)$, then one can consider the corresponding set of projective points in $\PG(m^2-1,q)$ under the canonical homomorphism $\psi:\GL(V(m^2,q))\rightarrow \PGL(m^2,q)$.  We prove (Theorem \ref{th_5}) that the projective set defined by $\cF_{m,q;I}$, with $|I|=k$, is a subset of a Desarguesian $m$-spread of $\PG(m^2-1,q)$ \cite{segre} consisting of two spread elements,   $k$ pairwise disjoint Segre varieties $\cS_{m,m}(\Fq)$ \cite{ht} and $q-1-k$ hyperreguli \cite{ost}.
 Additionally, if one consider the projective space $\PG(m^2-1,q)$ as the field reduction of $\PG(m-1,q^m)$ over $\Fq$, then the projective set defined by $\cF_{m,q;I}$ is, in fact, the field reduction of the union of two projective points,  $k$  mutually disjoint $(m-1)$-dimensional $\Fq$-subgeometries and $q-1-k$ scattered $\Fq$-linear sets of pseudoregulus type  of $\PG(m-1,q^m)$ \cite{dd,lvdv,mpt}. The main tool we use to get the above geometric description is the field reduction of $V(m,q^m)$ over $\Fq$ in the cyclic model for the tensor product $\Fqm\otimes V$   as described in \cite{coop}. 


\comment{In addition to  the applications  of the cyclic model for a  vector space over  $\Fq$ contained in this note, the authors strongly believe that  cyclic representation of (field reduction of) finite geometries is a useful tool to study other geometric objects as $\Fq$-linear sets and semifields and that the equivalence among them can be managed through Dickson matrix  representation of endomorphisms.
}

%
%
%
%
%

\section{The   non-linear MRD codes in the cyclic model of bilinear forms}\label{sec_3}

In the paper \cite{coop}, the cyclic model of the $m$-dimensional vector space $V=V(m,q)$  over $\Fq$ was introduced by taking  eigenvectors, say $v_1,\ldots, v_m$, of a given Singer cycle $\sigma$ of $V$, where a {\em Singer cycle} of $V$ is an element of $\GL(V)$ of order $q^m-1$.  Since the vectors $v_1,\ldots, v_m$ have distinct eigenvalues over $\Fqm$, they form a basis of  the extension $ \widehat V=V(m,q^m)$ of $V$.  In this basis  the vector space $V$ is represented by  
\begin{equation}\label{eq_53}
V=\left\{\sum_{j=1}^{m}{a^{q^{j-1}}v_j}:a \in \Fqm\right\}.
\end{equation}
We call $v_1,\ldots, v_m$ a {\em Singer basis} of $V$ and the above representation is  called the {\em cyclic model for} $V$ \cite{h1,fkmp}. 

 The set of all  $1-$dimensional $\Fq-$subspaces of $\widehat V$ spanned by  vectors in the cyclic model for $V$ is called the {\em cyclic model for  the projective space} $\PG(V)$. Note that the  above cyclic model  corresponds to the cyclic model of $\PG(V)$  where the points are identified with the elements of the group $\bZ_{q^{m-1}+q^{m-2}+\cdots+q+1}$ \cite[pp. 95--98]{h1}  \cite{fkmp}. Very recently, the cyclic model for $V(3,q)$ has been used to give an alternative model for  the  triality quadric $Q^+(7,q)$ \cite{bl}.

Let $\widehat V^*$ be the dual vector space of $\widehat V$ with basis 
 $v^*_1,\ldots,v^*_m$,  the dual basis of the Singer basis $v_1,\ldots,v_m$.
Then the dual vector space of $V$ is
\[
V^*=\left\{\sum_{i=1}^{m}{\alpha^{q^{i-1}} v_i^*}:\alpha\in\Fqm\right\}.
\]

A linear transformation from $V$ to itself is called an {\em endomorphism} of $V$. We will denote the set of all endomorphisms of $V$ by $\End(V)$.

  An $m\times m$ {\em Dickson matrix} (or $q${\em-circulant matrix}) over $\Fqm$ is a matrix of the form
 \[
D_{(a_0,a_1,\ldots,a_{m-1})}= \begin{pmatrix}
a_0           & a_1           & \cdots  & a_{m-1}\\
a^q_{m-1}     & a^q_{0}       & \cdots  & a^q_{m-2}\\
\vdots        & \vdots        &  \ddots & \vdots\\
a^{q^{m-1}}_{1} & a^{q^{m-1}}_{2} & \cdots  & a^{q^{m-1}}_0
\end{pmatrix}
 \]
with $a_i\in \Fqm$. We say that the above matrix is {\em generated  by the array} $(a_0,a_1,\ldots,a_{m-1})$.

Let $\cD_m(\Fqm)$ denote the {\em Dickson matrix algebra} formed by all $m\times m$ Dickson matrices over $\Fqm$. The set $\cB_m(\Fqm)$ of all invertible Dickson $m\times m$ matrices is known as the {\em Betti-Mathieu group} \cite{car}.

\begin{proposition}\label{prop_3}\cite[Lemma 4.1]{wl}
 $\End(V)\simeq\cD_m(\Fqm)$ and    $\GL(V)\simeq \cB_m(\Fqm)$.

\end{proposition}
%
%

A polynomial  of the form
\[
L(x)=\sum_{i=0}^{m-1}{\alpha_ix^{q^i}},\ \ \ \ \alpha_i\in\Fqm,
\]
 is called a {\em linearized polynomial} (or {\em q-polynomial}) over $\Fqm$. It is known that every endomorphism  of $\Fqm$ over $\Fq$ can be represented by a unique $q-$polynomial \cite{roman}.

Let $\cL_m(\Fqm)$ be the set of  all  $q$-polynomials over $\Fqm$. 
In the paper \cite{wl}, it was showed that  the map
\[
\begin{array}{cccc}
\varphi: & \cL_m(\Fqm) & \longrightarrow & \cD_m(\Fqm)\\
  &\sum_{i=0}^{m-1}{\alpha_ix^{q^i}} & \longmapsto & D_{(\alpha_0,\ldots ,\alpha_{m-1})}
\end{array}
\]
is an isomorphism between the non-commutative $\Fq-$algebras $\cL_m(\Fqm)$ and $\cD_m(\Fqm)$.   
 From Proposition \ref{prop_3} we see that any Singer basis  of $V$  realizes this  isomorphism.

\begin{proposition}\label{prop_12}
Let $v_1,\ldots,v_n$ be a Singer basis of $V$. Then the matrix of any $f\in\Omega$ with respect to $v_1,\ldots,v_n$ is  an $m\times m$  Dickson matrix. Conversely, every  $m\times m$ Dickson matrix defines a bilinear form on $V\times V$.
\end{proposition}
  \begin{proof}
Let $D_\a$ be an $m\times m$ Dickson  matrix generated by the $m$-ple $\a=(a_0,a_1,\ldots,a_{m-1})$ over $\Fqm$. Let $f_\a$  be the bilinear mapping on $ \widehat V\times  \widehat V$ defined by
\[
f_\a(v_i,v_j)=a_{m-i+j}^{q^{i-1}} \ \ \  \mathrm{for \ }i,j=1,\ldots,m
\]
where subscripts are taken modulo $m$, 
and then extended over $\widehat V$ by linearity.  Set  $L_{\a}(x)=\sum_{i=0}^{m-1}{a_ix^{q^i}}$ and let  $\Tr$ denote  the trace function from $\Fqm$ onto $\Fq$: 
\[
\Tr:y\in \Fqm \rightarrow \Tr(y)=\sum_{j=0}^{m-1}{y^{q^j}}\in \Fq.
\]

It is easily seen that the action  of $f_\a$ on $V\times V$ is given by
\begin{equation}\label{eq_2}
f_{\a}(v,v')=f_\a(x,x')=\Tr(L_{\a}(x')x),
\end{equation}
with 
$v=\sum_{i=1}^{m}{x^{q^{i-1}}}v_i, v'=\sum_{j=1}^{m}{x'^{q^{j-1}}}v_j$,    
which is a  bilinear form on $V\times V$.  The assertion follows from consideration on the size of $\cD_m(\Fqm)$. 
\end{proof}

For any  $m$-ple $\a=(a_0,\ldots,a_{m-1})$ over $\Fqm$, $f_{\a}$ will denote the bilinear form  having matrix $D_\a$ in the Singer basis $v_1,\ldots, v_m$. For any set $\cA$ of $m-$ples over $\Fqm$ we put 
\[
\cF_\cA=\{f_\a\in\Omega:\a\in \cA\}.
\]

\begin{corollary}\label{cor_3}
Let $\a=(a_0,\ldots, a_{m-1})$. Then
\begin{equation}\label{eq_60} 
\begin{array}{rccc}
\nu_{\{v_1,\ldots, v_{m}\}}:& \Omega & \rightarrow &  \cD_m(\Fqm)\\
 & f_\a & \mapsto  & D_{(a_0,\ldots, a_{m-1})}
\end{array}
\end{equation}
is an isomorphism of rank metric spaces giving $\rank(f_\a)=\rank(D_{(a_0,\ldots, a_{m-1})})$.
\end{corollary}
\begin{remark}
By Proposition \ref{prop_3}, $\Aut_{\Fq}(\Omega)$  is represented  by the group $(\cB_m(\Fqm)\times\cB_m(\Fqm))\rtimes\<t\>\rtimes\Aut(\Fq)$ in the  Singer basis $v_1,\ldots,v_m$. Here, $t$ denote transposition in $M_{m,m}(\Fqm)$ and it  corresponds to the operator $\top$.
\end{remark}

\begin{remark}\label{rem_5}
Note that $(\ref{eq_2})$ coincides with the bilinear form $(6.1)$ in $\cite{ds}$ when $m'=m$.
\end{remark}
%
%


%
\begin{remark}\label{rem_3}
Since a change of basis in $ \widehat V\times  \widehat V$  preserves the rank of  bilinear forms, for any given $f\in\Omega$ we can consider its matrix representation in the Singer basis $v_1,\ldots,v_m$. Therefore, we can assume $f=f_\a$ for some $m$-ple $\a$ over $\Fqm$, so that $\Rad(f_\a)$ is the set of vectors $v'=x'v_1+\ldots+x'^{q^{m-1}}v_m\in V$, $x'\in\Fqm$, such that   $L_{\a}(x')=0$.
\end{remark}

We are now in position to construct non-linear MRD codes as subsets of $\Omega$. 

Let $N$ denote the norm map from $\Fqm$ onto $\Fq$: 
\[
N:x\in \Fqm \rightarrow N(x)=\prod_{j=0}^{m-1}{x^{q^j}}\in \Fq.
\]

For every nonzero element $\alpha\in \Fqm$,  let 
\[
\pi_\alpha=\{(\lambda x,\lambda\alpha x^q,\lambda\alpha^{1+q} x^{q^2},\ldots,\lambda\alpha^{1+q+\ldots+q^{m-2}}x^{q^{m-1}}):\lambda,x \in\Fqm\setminus\{0\}\}.
\]

\begin{remark} \label{rem_2} 
 The matrix of the  Singer cycle $\sigma$ of $V$ in the basis $v_1,\ldots, v_m$ is  $\diag(\mu,\mu^q,\ldots,\mu^{q^{m-1}})$, where $\mu$ is a generator of the multiplicative group of $\Fqm$ $\cite{coop}$. If $S$ is the Singer cyclic group generated by $\sigma$, then the set $\cF_{\pi_a}$ is the $(S\times S)$-orbit of the bilinear form $f_\a$, with $\a=(1,\alpha,\alpha^{1+q},\ldots,\alpha^{1+\ldots+q^{m-2}})$. 
It turns out  that the bilinear forms in $\cF_{\pi_a}$ have constant rank.
\end{remark}

\begin{proposition}\label{prop_13}
 $\pi_\alpha=\pi_\beta$  if and only if $N(\alpha)=N(\beta)$. 
\end{proposition}
\begin{proof}
Let $\alpha,\beta\in\Fqm\setminus\{0\}$ such that  $N(\alpha)=N(\beta)$. By Remark \ref{rem_2} it suffices to show 
that $(1,\alpha,\alpha^{1+q}, \ldots,\alpha^{1+q+\ldots+q^{m-2}})$ is in  $\pi_\beta$. 

Since  $N(\alpha)=N(\beta)$, then  $\alpha=\beta c^{q-1}$ for some $c\in\Fqm\setminus\{0\}$. As \\
$(1+q+\ldots+q^{k})(q-1)=q^{k+1}-1$, we have
\[
\alpha^{1+q+\ldots+q^{k}} = c^{-1}\beta^{1+q+\ldots+q^{k}}  c^{q^{k+1}}.
\]

Conversely, let $\pi_\alpha=\pi_\beta$. Then
\begin{equation}\label{eq_14}
\begin{array}{rcl}
1 & = & \lambda x\\
\alpha  & =  &  \lambda \beta   x^{q}\\
\alpha^{1+\ldots+q^{m-2}} & =  &  \lambda \beta^{1+\ldots+q^{m-2}}   x^{q^{m-1}}
\end{array} 
\end{equation}
for   some $\lambda,x\in\Fqm\setminus\{0\}$.
From the last equation we get 
\[
\alpha^{q+q^2+\ldots+q^{m-1}} =   \lambda^q\beta^{q+q^2+\ldots+q^{m-1}} x.
\]
By taking into account the first and second   equation of (\ref{eq_14}) we get  
\[
N(\alpha)=\lambda^q\lambda N(\beta) xx^q=N(\beta).
\]
\end{proof}
We will write $\pi_a$ instead of $\pi_\alpha$,  if $\alpha$ is an element of $\Fqm\setminus\{0\}$ with $N(\alpha)=a$. 

\begin{lemma}\label{lem_6}
Every $\pi_a$ has size $(q^m-1)^2/(q-1)$.
\end{lemma}
\begin{proof}
Let $\alpha\in\Fqm\setminus\{0\}$ with $N(\alpha)=a$.
Clearly, we  have
\[
(\lambda x,\lambda\alpha x^q,\lambda\alpha^{1+q} x^{q^2},\ldots,\lambda\alpha^{1+\ldots+q^{m-2}}x^{q^{m-1}})
=(\rho y,\rho\alpha y^q,\rho\alpha^{1+q} y^{q^2},\ldots,\rho\alpha^{1+\ldots+q^{m-2}}x^{q^{m-1}})
\]
if and only if $\lambda x^{q^i}=\rho y^{q^i}$, for $i=0,\ldots,m-1$. If we compare the equalities  with $i=0$ and $i=1$,  we get $x^{q-1}=y^{q-1}$. For every fixed $x\in\Fqm$ there are exactly $q-1$ elements $y$ in $\Fqm$ such that $y^{q-1}=x^{q-1}$.  

Let   $\lambda$ and $x$ be fixed elements in $\Fqm\setminus\{0\}$. Then,   for each  element $y\in\Fqm$  such that $y^{q-1}=x^{q-1}$ we get the unique element $\rho=\lambda x y^{-1}$ and the result is proved.  
\end{proof}

\begin{lemma}\label{prop_9}
\begin{itemize}
\item[i)] If $\a\in \pi_1$, then  $\rank(f_\a)=1$. 
\item[ii)]If $a,b\in\Fq\setminus\{0,1\}$,  then  $\rank(f_\a-f_{\b})\ge m-1$, for any $\a\in \pi_a$ and  $\b\in\pi_b$, with $\b\neq \a$ if $a=b$. 
\end{itemize}
\end{lemma}
\begin{proof}
i) Let $\a=(\lambda x,\lambda x^q,\ldots, \lambda x^{q^{m-1}})\in\pi_1$. It suffices to  note that $L_\a(z)=(\lambda x)z+(\lambda x^q)z^q+\ldots (\lambda x^{q^{m-1}})z^{q^{m-1}}=0$ is the equation of a hyperplane in the cyclic model of $V$.

ii) 
 By  Remark \ref{rem_2}, we can assume $\a=(1,\alpha,\ldots,a^{1+\ldots+q^{m-2}})$, with $N(\alpha)=a\neq 1$.
 
Let $\b=(\lambda x,\lambda \beta x^q, \ldots, \lambda \beta^{1+q+\ldots+q^{m-2}}x^{q^{m-1}})$, with $N(\beta)=b\neq 1$.

Suppose there exist $z_1,z_2\in\Fqm$ linearly independent over $\Fq$ such that $L_{\a-\b}(z_i)=0$. Then we get
\begin{equation}\label{eq_28}
\begin{array}{lcl}
(1-\lambda x)z_i+(\alpha-\lambda\beta x^q) z_i^q+\ldots +(\alpha^{1+\ldots+q^{m-2}}-\lambda\beta^{1+\ldots+q^{m-2}}x^{q^{m-1}})z_i^{q^{m-1}} &= & 0
\end{array}
\end{equation}
and
\begin{equation}\label{eq_29}
\begin{array}{rcl}
(\alpha^{q+\ldots+q^{m-1}}-\lambda^q\beta^{q+\ldots+q^{m-1}}x)z_i+(1-\lambda^q x^q)z_i^q+\\
[.05in]
\ldots +(\alpha^{q+\ldots+q^{m-2}}-\lambda^q\beta^{q+\ldots+q^{m-2}}x^{q^{m-1}})z_i^{q^{m-1}} & = & 0,
\end{array}
\end{equation}
for $i=1,2$. 

After subtracting Equation (\ref{eq_28})  side-by-side  from  Equation (\ref{eq_29}) multiplied by $\alpha$, we get

\begin{equation}\label{eq_31}
\begin{array}{rcl}
[a-1+(\lambda -\lambda^q\alpha\beta^{q+\ldots+q^{m-1}})x]z_i
+(\lambda\beta -\lambda^q\alpha )x^qz_i^q+\\[.05in]
\ldots
+(\lambda\beta-\lambda^q\alpha)\beta^{q+\ldots+q^{m-2}} x^{q^{m-1}}z_i^{q^{m-1}} & = & 0,
\end{array}
\end{equation}
for $i=1,2$. Then, the $m-$ple 
 \begin{equation}\label{eq_32}
 (a-1+(\lambda -\lambda^q\alpha\beta^{q+\ldots+q^{m-1}})x,
(\lambda\beta -\lambda^q\alpha )x^q,\ldots,
(\lambda\beta-\lambda^q\alpha)\beta^{q+\ldots+q^{m-2}}x^{q^{m-1}})
 \end{equation}
 is a solution of the linear system
\begin{equation}\label{eq_34}
\left\{
\begin{array}{lcl}
z_1X_1+z_1^q X_2+\ldots+z_1^{q^{m-1}} X_m & = & 0\\
z_2X_1+z_2^q X_2+\ldots+z_2^{q^{m-1}} X_m & = & 0
\end{array}
\right.
\end{equation}
with $\Delta=\begin{vmatrix} z_1 & z_1^q\\z_2 & z_2^q\end{vmatrix}\neq 0$. 

The generic solution  $(x_1,x_2,\ldots,x_m)$ of (\ref{eq_34}) has 
\comment{
Let us write the linear system (\ref{eq_34}) as 
\begin{equation*}
\left\{
\begin{array}{lcl}
z_1X_1+z_1^q X_2 & = & -z_1^{q^2}X_3-\ldots-z_1^{q^{m-1}} X_m \\
z_2X_1+z_2^q X_2& = & -z_2^{q^2}X_3-\ldots-z_2^{q^{m-1}} X_m 
\end{array}
\right.
\end{equation*}
so that we get
}
\begin{equation}\label{eq_63}
x_1  =  -\Delta^{-1}\left(\begin{vmatrix} z_1^{q^2} & z_1^q\\z_2^{q^2} & z_2^q\end{vmatrix}x_3+
         \begin{vmatrix} z_1^{q^3} & z_1^q\\z_2^{q^3} & z_2^q\end{vmatrix}x_4+\ldots
          +\begin{vmatrix} z_1^{q^{m-1}} & z_1^q\\z_2^{q^{m-1}} & z_2^q\end{vmatrix}x_m\right)
\end{equation}
and 
\begin{equation}\label{eq_62}
x_2  =  -\Delta^{-1}\left(\begin{vmatrix} z_1 & z_1^{q^2}\\z_2 & z_2^{q^2}\end{vmatrix}x_3+
         \begin{vmatrix} z_1 & z_1^{q^3}\\z_2 & z_2^{q^3}\end{vmatrix}x_4+\ldots
          +\begin{vmatrix} z_1 & z_1^{q^{m-1}}\\z_2 & z_2^{q^{m-1}}\end{vmatrix}x_m\right).
\end{equation}

In the expression (\ref{eq_62})  set 
\[
c_i=\begin{vmatrix} z_1 & z_1^{q^{i-1}}\\z_2 & z_2^{q^{i-1}}\end{vmatrix}, \ \ \ \ i=3,\ldots,m;
\] 
in particular $c_m=\begin{vmatrix} z_1 & z_1^{q^{m-1}}\\z_2 & z_2^{q^{m-1}}\end{vmatrix}$ giving $c_m^q=-\Delta$. Similarly, in the expression (\ref{eq_63}) set 
\[
d_i=
       -\begin{vmatrix}   z_1 & z_1^{q^{i-2}}\\  z_2 & z_2^{q^{i-2}} \end{vmatrix}^q, \ \ \ \ i=3,\ldots,m.
\] 
We have, 
\begin{equation*}
d_i=-c_{i-1}^q,\ \ \ \ \mathrm{for \ } i=4,\ldots,m 
\end{equation*}
and $d_3=-\Delta^q$. We then write
\[
\begin{array}{lcl}
x_1 & = & -\Delta^{-1}(-\Delta^q x_3-c_3^qx_4-\ldots-c_{m-1}^qx_m)\\[.01in]
x_2 & = & -\Delta^{-1}(c_3 x_3+c_4x_4+\ldots+c_{m}x_m)\\
\end{array}
\]

By plugging  (\ref{eq_32}) in the right-hands of the above equalities  
\comment{
of the linear system (\ref{eq_34}), we have 
\[
\begin{array}{lcl}
x_1 & = & N(\alpha)-1+(\lambda-\lambda^q\alpha\beta^{q+\ldots+q^{m-1}}) x\\[.01in]
x_2 & = & (\lambda\beta -\lambda^q\alpha )x^q\\[.01in]
x_3 & = & (\lambda\beta -\lambda^q\alpha)\beta^{q}x^{q^2}\\[.01in] 
& \vdots & \\[.01in]
x_m & = & (\lambda\beta-\lambda^q\alpha)\beta^{q+\ldots+q^{m-2}}x^{q^{m-1}}.
\end{array}
\]

By plugging the expression of $x_3,\ldots, x_m$ in equations (\ref{eq_38}),
} we get
\begin{equation*}
-\Delta^q x_3-c_3^qx_4-\ldots-c_{m-1}^qx_m=(\lambda\beta -\lambda^q\alpha)(-\Delta^q\beta^{q}x^{q^2}-c_3^q\beta^{q+q^2}x^{q^3}-\ldots-c_{m-1}^q\beta^{q+\ldots+q^{m-2}}x^{q^{m-1}})
\end{equation*}
and
\begin{equation*}
c_3 x_3+c_4x_4+\ldots+c_{m}x_m=(\lambda\beta -\lambda^q\alpha)(c_3\beta^{q}x^{q^2}+c_4\beta^{q+q^2}x^{q^3}+\ldots+c_{m}\beta^{q+\ldots+q^{m-2}}x^{q^{m-1}}).
\end{equation*}
\comment{
Thus, we have
\begin{equation*}
\begin{array}{lcl}
\dfrac{-\Delta x_1}{\lambda\beta -\lambda^q\alpha} & = & -\Delta^q \beta^{q}x^{q^2}- c_3^q\beta^{q+q^2}x^{q^3} -\ldots-c_{m-1}^q\beta^{q^2+\ldots+q^{m-2}}x^{q^{m-1}}
\end{array}
\end{equation*}
and 
\begin{equation*}
\begin{array}{lcl}
\left(\dfrac{-\Delta x_2}{\lambda\beta -\lambda^q\alpha}\right)^q & = & c_3^q\beta^{q^2}x^{q^3} +c_4^q \beta^{q^3+q^2}x^{q^4} +\ldots+c_{m}^q\beta^{q^2+\ldots+q^{m-1}}x^{q^{m}}\\ [0.1in]
  & = & -\Delta \beta^{q^2+\ldots+q^{m-1}}x+ c_3^q\beta^{q^2}x^{q^3} +\ldots+c_{m-1}^q\beta^{q^2+\ldots+q^{m-2}}x^{q^{m-1}}
\end{array}
\end{equation*}
giving 
\begin{equation*}
\begin{array}{lcl}
\beta^{q}\left(\dfrac{-\Delta x_2}{\lambda\beta -\lambda^q\alpha}\right)^q & =  & -\Delta \beta^{q+q^2+\ldots+q^{m-1}}x+ c_3^q\beta^{q+q^2}x^{q^3} +\ldots+c_{m-1}^q\beta^{q+q^2+\ldots+q^{m-2}}x^{q^{m-1}}.
\end{array}
\end{equation*}
}
Therefore
\begin{equation}\label{eq_42}
\begin{array}{lcl}
\beta^{q}\left(\dfrac{-\Delta x_2}{\lambda\beta -\lambda^q\alpha}\right)^q +\dfrac{-\Delta x_1}{\lambda\beta -\lambda^q\alpha}& =  & -\Delta \beta^{q+q^2+\ldots+q^{m-1}}x -\Delta^q x^{q^2}\beta^q.
\end{array}
\end{equation}

From (\ref{eq_32}), we have $x_2=(\lambda\beta -\lambda^q\alpha)x^q$ giving 
\[
\beta^{q}\left(\dfrac{-\Delta x_2}{\lambda\beta -\lambda^q\alpha}\right)^q=(-1)^q\Delta^qx^{q^2}\beta^q.
\]

From (\ref{eq_42})
it turns out that the value of $x_1$ must satisfy
\[
-\dfrac{\Delta x_1}{\lambda\beta -\lambda^q\alpha} =  -\Delta \beta^{q+q^2+\ldots+q^{m-1}}x
\]
giving
\begin{equation*}
x_1=(\lambda\beta -\lambda^q\alpha) \beta^{q+q^2+\ldots+q^{m-1}}x=(\lambda b-\lambda^q\alpha \beta^{q+q^2+\ldots+q^{m-1}})x
\end{equation*}
since $\Delta\neq 0$.

From (\ref{eq_32}), we have $x_1=(a-1)+(\lambda-\lambda^q\alpha\beta^{q+\ldots+q^{m-1}}) x$. Therefore,  we get 
\[
(b-1)\lambda x = a-1
\]
i.e.,
\begin{equation}\label{eq_43}
\lambda =\dfrac{a-1}{b-1}x^{-1}.
\end{equation}

By plugging this value in  $\b$, we get
\[
\b=\dfrac{a-1}{b-1}\left(1,\beta x^{q-1},\beta^{1+q} x^{q^2-1},\ldots,\beta^{1+q+\ldots+q^{m-2}} x^{q^{m-1}-1}\right)
\]

Note that if $b=a$, we can assume $\beta=\alpha$ giving $x\not\in\Fq$ as $\b\neq \a$.

We claim that the bilinear form $(f_\a-f_\b)$ has maximum rank $m$. Indeed, suppose there exists a nonzero $z\in\Fqm
$ such that $L_{\a-\b}(z)=0$.  By plugging (\ref{eq_43}) in Equation (\ref{eq_31}) we get 
\[
\begin{array}{rcl}
\dfrac{a-1}{b-1}\left[(\beta-\alpha(x^{-1})^{q-1})\beta^{q+\ldots+q^{m-1}}z+(\beta x^{q-1}-\alpha)z^q+\right. &&  \\[.05in]
 \ldots+ \left.(\beta x^{q^{m-1}-1}-\alpha x^{q^{m-1}-q})\beta^{q+\ldots+q^{m-2}}z^{q^{m-1}}\right] & = &0
\end{array}
\]
or, equivalently,
\[
\begin{array}{lcl}


\left(\dfrac{\beta}{x}-\dfrac{\alpha}{x^q}\right)(\beta^{q+\ldots+q^{m-1}}xz
+(xz)^q+\beta^q(xz)^{q^2}+\ldots+\beta^{q+\ldots+q^{m-2}}(xz)^{q^{m-1}}) & =&0,

\end{array}
\]
where $\dfrac{\beta}{x}-\dfrac{\alpha}{x^q} \neq0$ since either $b\neq a$ or $x^q\neq x$ if $b=a$. Therefore,  the  following equation holds:
\begin{equation}\label{eq_45}
\begin{array}{lcl}
\beta^{q+\ldots+q^{m-1}}y
+y^q+\beta^qy^{q^2}+\beta^{q+q^2}y^{q^3}+\ldots+\beta^{q+\ldots+q^{m-2}}y^{q^{m-1}} & = &0
\end{array}
\end{equation}
given
\begin{equation}\label{eq_46}
\begin{array}{lcl}
\beta^{q^2+\ldots+q^{m-1}}y+\beta^{1+q^2+\ldots+q^{m-1}}y^q
+y^{q^2}+\beta^{q^2}y^{q^3}+\ldots+\beta^{q^2+\ldots+q^{m-2}}y^{q^{m-1}} & = &0.
\end{array}
\end{equation}

By subtracting Equation (\ref{eq_45}) from (\ref{eq_46}) multiplied by $\beta^q$  we get $b=1$, a contradiction. 
\end{proof}

For every nonzero element $\alpha\in\Fqm$, let  
\[
J_\alpha=\{(\lambda x,0,\ldots,0,-\lambda\alpha x^{q^{m-1}}):\lambda,x\in\Fqm\setminus\{0\}\}.
\]

\begin{remark} \label{rem_4} 
 Note that the set $\cF_{J_\alpha}$ is the $(S\times S)$-orbit of the bilinear form $f_\a$, with $\a=(1,0,\ldots,0,-\alpha)$. 
It turns out  that the bilinear forms in $\cF_{J_\alpha}$ have constant rank.
\end{remark}

By arguing similarly to  the proof of Proposition \ref{prop_13} and Lemma \ref{lem_6}, we get the following result.

\begin{lemma}\label{prop9a}

Each set $J_\alpha$ has size $(q^m-1)^2/(q-1)$ and 
$J_\alpha=J_\beta$ if and only if $N(\alpha) = N(\beta)$. 
\end{lemma}

\comment{
\begin{proof}
Just mimic  the proof of Proposition \ref{}.

Let $\alpha,\beta\in\Fqm\setminus\{0\}$ such that  $N(\alpha)=N(\beta)$. It suffices to show that $(1,0,0, \ldots,0,-\alpha)$ is in  $J_\beta$. Since  $N(\alpha)=N(\beta)$ then  $\alpha=\beta c^{q-1}$ for some $c\in\Fqm\setminus\{0\}$.  It follows that $\alpha = c^{-1}\beta c^q=c^{-1}\beta d^{q^{m-1}}$ for some nonzero $d\in \Fq^m$.

Conversely, let $J_\alpha=J_\beta$. Then
\begin{equation}\label{eq_14a}
\begin{array}{rcl}
1 & = & \lambda x\\
\alpha  & =  &  \lambda \beta   x^{q^{m-1}}\\
\end{array} 
\end{equation}
for   some $\lambda,x\in\Fqm\setminus\{0\}$.

It follows that $\alpha = \beta   x^{q^{m-1}-1}$ hence $N(\alpha)=N(\beta).$

\end{proof}
}
We will write $J_a$ instead of $J_\alpha$,  if $\alpha$ is an element of $\Fqm$ with $N(\alpha)=a$.

\begin{lemma}\label{prop_11}
For any $\a=(x,0,\ldots,0,y)$ with $x,y\in\Fqm$ not both zero,  $\rank(f_\a)\ge m-1$. 
\end{lemma}
\begin{proof}
The  bilinear form $f_\a$, is equivalent to the bilinear form $f_{\hat \a}$, with $\hat\a=(x,y^q,0,\ldots,0)$, via the automorphism $\top$.
The result then follows from Remark \ref{rem_5} and Theorem 6.3 in \cite{ds}. 
\end{proof}
\begin{corollary}\label{prop_11a}
Let $a,b$ be nonzero elements in $\Fq$. Then $\rank(f_\a-f_{\b})\ge m-1$,  for any $\a\in J_a$ and $\b\in J_b$, with $\a\neq \b$ if a=b.
\end{corollary}
\begin{lemma}\label{lem_5a}
 Let $a,b$ be distinct nonzero elements in $\Fq$. Then $\rank(f_{\a}-f_{\b})\ge m-1$ for any $\a\in\pi_a$ and $\b\in J_b$.
\end{lemma}
\begin{proof}
By Remark \ref{rem_2}  we can assume $\a=(1,\alpha,\ldots,\alpha^{1+\ldots+q^{m-2}})$ with $N(\alpha)=a$. By  arguing as in the proof of Lemma \ref{prop_9} we see that the triple 
\begin{equation}\label{eq_33}
(a-1+ (\lambda+ \alpha\beta^q \lambda^q)x,-\alpha\lambda^qx^q,-\lambda \beta x^{q^{m-1}})
\end{equation}
 is a solution of the linear system
\begin{equation}\label{eq_27}
\left\{
\begin{array}{lcl}
z_1X_1+z_1^q X_2+z_1^{q^{m-1}} X_3 & = & 0\\
z_2X_1+z_2^q X_2+z_2^{q^{m-1}} X_3 & = & 0
\end{array}
\right.
\end{equation}
for some $z_1,z_2\in\Fqm$ linearly independent over $\Fq $ with $\Delta=\begin{vmatrix} z_1 & z_1^q\\z_2 & z_2^q\end{vmatrix}\neq 0$. Any  solution $(x_1,x_2,x_3)$ of (\ref{eq_27}) satisfies 
\[
x_2=-\dfrac{\Delta'}{\Delta}x_3
\]
where $\Delta'=\begin{vmatrix} z_1 & z_1^{q^{m-1}}\\z_2 & z_2^{q^{m-1}}\end{vmatrix}$. Since ${\Delta'}^q=\begin{vmatrix} z_1^q & z_1\\z_2^q & z_2\end{vmatrix}=-\Delta$ we get
$x_2=\dfrac{1}{{\Delta'}^{q-1}}x_3$ giving $N(x_2)=N(x_3)$. As a solution of (\ref{eq_27}), the triple (\ref{eq_33}) must satisfies $aN(\lambda)N(x)=b N(\lambda) N(x)$ giving either $\lambda x=0$ or $a=b$, a contradiction.
\end{proof}
Let $A_1=\{(x,0,0,\ldots,0):x\in\Fqm\setminus\{0\}\}$ and   $A_2=\{(0,0,0,\ldots,x):x\in\Fqm\setminus\{0\}\}$.
\begin{lemma}\label{prop_10}
 $\rank(f_\a)=m$,  for any $\a\in A_i$, $i=1,2$. Further, $\rank(f_\a-f_{\b})\ge m-1$, for any $\a\in A_1$ and $\b\in A_2$.
\end{lemma}
\begin{proof}
The first part can be easily proved by taking the Dickson matrix $D_\a$ with $\a\in A_i$. The second part follows from  Lemma \ref{prop_11}.
\end{proof}
\begin{lemma}\label{lem_10}
Let $a\in\Fq\setminus\{0,1\}$. Then  $\rank(f_\a-f_{\b})\ge m-1$, for any $\a\in\pi_a$  and $\b\in A_i$, $i=1,2$. 
\end{lemma}
\begin{proof}
 By Remark \ref{rem_2} we can assume $\a=(1,\alpha,\ldots,\alpha^{1+\ldots+q^{m-2}})$ with $N(\alpha)=a$. Let $\b=(x,0,\ldots, 0)$. By  proceeding as in the proof of Lemma \ref{prop_9} we see the pair $(a-(1-x),-\alpha x^q)$ is a solution of the linear system
\[\left\{
\begin{array}{lcl}
z_1X_1+z_1^q X_2 & = & 0\\
z_2X_1+z_2^q X_2 & = & 0
\end{array}
\right.
\]
with  $\Delta=\begin{vmatrix} z_1 & z_1^q\\z_2 & z_2^q\end{vmatrix}\neq 0$. Then the above linear system has the unique solution $(0,0)$ giving $x=0$ and $a=1$, a contradiction.

For $i=2$, similar arguments  lead to the same contradiction.
\end{proof}
\begin{lemma}\label{lem_11}
Let $a\in\Fq\setminus\{0\}$. Then  $\rank(f_\a-f_{\b})\ge m-1$, for any $\a\in J_a$  and $\b\in A_i$, $i=1,2$. 
\end{lemma}
\begin{proof}
Use Lemma \ref{prop_11}.
\end{proof}

Finally, we have the main theorem.

\begin{theorem}\label{th_1}
Let $q>2$ be a prime power and $m\ge 3$ a positive integer. For any subset  $I$  of $\Fq\setminus\{0,1\}$, put  $\Pi_I=\bigcup_{a\in I}{\pi_a}$, $\Gamma_I=\bigcup_{b\in \Fq \setminus (I\cup\{0\})}{J_b}$ and set
\[
\cA_{m,q;I}=\Pi_I\cup\Gamma_I  \cup A_1\cup A_2\cup \{\0\}
\]
where $\0$ is the zero $m-$ple. Then the  subset  $\cF_{m,q;I}=\{f_\a:\a\in \cA_{m,q;I}\}$ of $\Omega$  is a non-linear $(m,m,q;m-2)$-MRD code. 
\end{theorem}

\begin{proof}
By Lemmas \ref{lem_6}, \ref{prop9a}  we get that $\cA_{m,q;I}$ has size $q^{2m}$. 
By Lemmas \ref{prop_9}, \ref{prop_11}, \ref{lem_5a}, \ref{prop_10}, \ref{lem_10} and Corollary \ref{prop_11a}, we see that  $\cF_{m,q;I}$  has minimum  distance  $m-1$, i.e. it is a $(m,m,q;m-2)$-MRD code. To show the non-linearity of $\cF_{m,q;I}$, it suffices to find  two distinct elements in it whose $\Fq$-span is not contained in $\cF_{m,q;I}$.

Let $f_\a\in\cF_{A_2}$ and $f_\b\in \cF_{\pi_a}$, $a\in I$. By corollary \ref{cor_3}, we can work with the Dickson matrices $D_\a$ and $D_\b$, or equivalently, with $m$-ples $\a$ and $\b$ as arrays in $V(m,q^m)$. Let $\a=(0,\ldots,0,\mu)$ and $\b=(\lambda x,\lambda \alpha x^q,\ldots,\lambda \alpha^{1+\ldots+q^{m-2}}x^{q^{m-1}})$. Suppose $\a+\b\in\pi_b$, for some  $b\in\Fq$. Then 
\[
\left(\frac{\lambda \alpha^{1+\ldots+q^{m-3}}x^{q^{m-2}}}{\lambda \alpha^{1+\ldots+q^{m-4}}x^{q^{m-3}}}\right)^q= \alpha^{q^{m-2}}x^{q^{m-1}-q^{m-2}}=\frac{\mu+\lambda \alpha^{1+\ldots+q^{m-2}}x^{q^{m-1}}}{\lambda \alpha^{1+\ldots+q^{m-3}}x^{q^{m-2}}}
\]
 giving $\mu=0$. Therefore, the subspace spanned by $\a$ and $\b$ meets trivially every $\pi_b$ if $b\neq a$, or just in the 1-dimensional subspace spanned by $\b$ if $b=a$. The result then follows.
\end{proof}

%
%
%
%

\section{A geometric description for the  non-linear  MRD codes}\label{sec_4}

For any  $v\in V(t,q^s)\setminus\{0\}$,  $[v]$ will denote the point of  $\PG(t-1,q^s)$ defined by $v$ via  the canonical homomorphism $\psi:\GL(V(t,q^s))\mapsto \PGL(t,q^s)$. For any subset $X$ of  $V(t,q^s)\setminus\{0\}$, we set $[X]=\{[v]:v \in X, v\neq 0\}$.   The set $[X]$ is said to be an  $\Fq${\em-linear set of rank} $r$ if $X$ is an $r$-dimensional $\Fq$-linear subspace of $V(t,q^s)$. An $\Fq$-linear set $[X]$ of rank $r$ is said to be  {\em scattered} if the size of $[X]$ equals $|\PG(r-1,q)|$; see \cite{pol} for more  details on $\Fq$-linear sets
and \cite{lun} for a relationship between linear MRD-codes and $\Fq$-linear sets. 

Consider the set $\cA_{m,q;I}$ defined in Theorem \ref{th_1} as a subset of $\widehat V=V(m,q^m)$, by setting $a_0v_1+a_1v_2+\ldots+a_{m-1}v_m$, for any $\a=(a_0,\ldots,a_{m-1})\in\cA_{m,q;I}$; here, $v_1,\ldots, v_m$ is the Singer basis of $V$ defined in Section \ref{sec_3}. Therefore,    $ [\pi_1]=[V]$ is a  scattered $\Fq$-linear set  of rank $m$ of $\PG(m-1,q^m)$  isomorphic to  the projective space $\PG(m-1,q)$. 

For any $\alpha\in\Fqm\setminus\{0\}$, the endomorphism 
\[
\begin{array}{rccc}
\tau_\alpha: & \widehat V & \rightarrow &  \widehat V\\
 & a_0v_1+a_1v_2+\ldots+a_{m-1}v_m & \mapsto  & a_0v_1+\alpha a_1v_2+\ldots+\alpha^{1+\ldots+q^{m-2}}a_{m-1}v_m
\end{array}
\]
maps $\pi_1$ into $\pi_a$, with $a=N(\alpha)$, and $J_1$ into $J_b$, with $b=a^{m-1}$. 
 
Let $W$ be the span of $v_1$ and $v_m$ in $\widehat V$. For any $a\in\Fq\setminus\{0\}$,  $[ J_a]$ is a scattered $\Fq$-linear set of rank $m$ of  $[W]$.
In particular  $[J_a]$ is a maximum scattered  $\Fq$-linear set of  pseudoregulus type of 
$[W]$ \cite{lvdv,mpt}.

Summarizing we have  the following result.

\begin{theorem}\label{th_4}
Let $q>2$ be a prime power and $m>2$ a positive integer. Let  $I$  be any nonempty subset of $\Fq\setminus\{0,1\}$ with $k=|I|$. Then, the projective image of  $\cA_{m,q;I}$ in $\PG(m-1,q^m)$  is  union of  two points $[A_1], [A_2]$,  $k$ mutually disjoint $(m-1)$-dimensional $\Fq$-subgeometries $[\pi_a]$, $a\in I$, and  $q-1-k$ mutually disjoint $\Fq$-linear sets $[J_b],  b \in \Fq\setminus (I\cup\{0\})$,  of pseudoregulus type of rank $m$ contained in the line  spanned by  $[A_1]$ and $ [A_2]$.
\end{theorem}

We now investigate the geometry in $\PG(m^2-1,q)$ of the projective set defined by each   MRD code  $\cF_{m,q;I}$ viewed as a subset of $V(m^2,q)$.

Let $V=V(m,q)$ be the $\Fq$-span of $u_1,\ldots, u_m$ and set $\widehat V=V(m,q^m)=\Fqm\otimes V(m,q)$.  The {\em rank} of a   vector $v=a_1  u_1+a_2  u_2+\ldots+a_{m}  u_m\in \widehat V$ by definition is the maximum number of linearly independent coordinates $a_i$ over $\Fq$.

 If we consider  $\Fqm$ as  the $m$-dimensional vector space $V$,  then  every  $\alpha\in\Fqm$ can be uniquely written as $\alpha=x_1  u_1+x_2  u_2+\ldots+x_{m}  u_m$, with $x_i\in\Fq$. Hence, $\widehat V$ can be viewed as $V\otimes V$, the tensor product of $V$ with itself, with basis $\{u_{(i,j)}=u_i\otimes u_j:i,j=1,\ldots,m\}$. Elements of $V\otimes V$  are called {\em tensors} and those of  the form $v\otimes v'$, with $v,v' \in V$ are called {\em fundamental tensors}. In $\PG(V\otimes V)$, the set of fundamental tensors correspond to the Segre variety $\cS_{m,m}(\Fq)$ of $\PG(V\otimes V)$ \cite{ht}.  

Let $\phi$ be the map defined by 

\[
\begin{array}{lccc}
\phi=\phi_{\{u_1,\ldots, u_m\}}: & \widehat V & \longrightarrow & V\otimes V \\
 & \alpha_1 u_1+\ldots +\alpha_{m} u_m & \longmapsto  & \sum_{i=1}^{m}{x_{i1}\,u_{(i,1)}}+\ldots+\sum_{i=1}^{m}{x_{im}\,u_{(i,m)}},
 \end{array}
\]
with $\alpha_k=x_{1k} u_1+x_{2k} u_2+\ldots+x_{mk} u_m$,  $x_{ik}\in\Fq$. We call this map  the {\em field reduction of }  $\widehat V$ {\em over} $\Fq$ {\em with respect to the  basis}  $u_1,\ldots, u_m$.  
  The projective space $\PG(V\otimes V)$   is  the {\em the field reduction} of $\PG(\widehat V)$ {\em over} $\Fq$ {\em with respect to the  basis} $u_1,\ldots, u_m$.

 Under the map $\phi$, every 1-dimensional   subspace $\<v\>$  of $\widehat V$ is mapped  to the $m$-dimensional  $\Fq$-subspace $k_{v}=\phi(\<v\>)$ of $V\otimes V$. It turns out that the set $\cK=\{k_v: v\in \widehat V, v\neq 0\}$ is 
 a partition of the nonzero vectors  of $V\otimes V$.
 In particular $\cK$ is a  {\em Desarguesian} partition, i.e.  the stabilizer of $\cK$ in $\GL(V\otimes V)$  contains a cyclic subgroup acting regularly on the components of $\cK$ \cite{segre}, \cite{dye}.


To any component  $k_v$ of $\cK$ there corresponds a projective $(m-1)-$dimensional subspace $[k_v]$ of $\PG(V\otimes V)$. The set  $\cS=\{[k_v]: v\in \widehat V, v\neq 0\}$  is so called a {\em Desarguesian  $(m-1)-$spread} of $\PG(V\otimes V)$ \cite{segre}, \cite{dye}.

In addition, the projective set of $\PG(V\otimes V)$ corresponding to the  $\phi$-image of the  1-dimensional subspaces spanned by non-zero vectors in $V$ is the Segre variety $\cS_{m,m}(\Fq)$.

Let $\nu$ be the map  defined by 
\[
\begin{array}{lccc}
\nu=\nu_{\{u_1,\ldots, u_m\}}:&  V\otimes  V & \longrightarrow & M_{m,m}(\Fq)\\
&  \sum_{i,j}{x_{ij}u_{(i,j)}}  & \longrightarrow   & (x_{ij})_{i,j=1,\ldots,m}.
 \end{array}
\]

For every $v=\alpha_1 u_1+\ldots+\alpha_{m} u_m\in \widehat V$,  the $k$-th column of the matrix 
 $\nu(\phi(v))$   is the $m$-ple $(x_{1k}, \ldots,x_{mk})$ of the coordinates of  $\alpha_{k}$ with respect to the basis $u_1,\ldots,u_m$ of $\Fqm$. From \cite{gab}, the rank of $v$ equals the rank of  $\nu(\phi(v))$, for all $v\in\widehat V$.  In addition, the $\nu$-image of fundamental tensors is precisely the set of   rank 1 matrices.

\begin{remark}
Evidently, $\nu$ is an isomorphism of rank metric  spaces which also provides  an isomorphism between  the field reduction  $V\otimes V$ of $\widehat V$ with respect to  $u_1,\ldots, u_m$ and  the metric space  $\Omega$ of all bilinear forms on $V=\<u_1,\ldots,u_m\>_{\Fq}$.  
\end{remark}

Now embed $V\otimes V$ into $\widehat V\otimes \widehat V$ by extending the scalars from $\Fq$ to $\Fqm$. By taking a Singer basis $v_1,\ldots,v_m$ of  $V$ defined by the Singer cycle $\sigma$, Cooperstein \cite{coop} defined a  cyclic model  
for $V\otimes V$ within $\widehat V\otimes \widehat V$ with basis $v_{(i,j)}=v_i\otimes v_j$, $i,j=1,\ldots, m$. Let 
\[
\Phi(j)=\{\sum_{i=1}^{m}{a^{q^{i-1}}v_{(i,j-1+i)}}:a \in \Fqm\},
\]
where the subscript $j-1+i$ is taken modulo $m$.
As an $\Fq$-space, $\Phi(j)$ has dimension $m$ and by consideration on dimension we have
\[
V\otimes V=\bigoplus_{j=1}^{m}{\Phi(j)};
\]
see \cite{coop}. We call this representation the {\em cyclic representation of the tensor product} $V\otimes V$.

\begin{proposition}
Let $\widetilde\phi$ be the map defined by
\[
\begin{array}{lccc}
\widetilde\phi=\phi_{\{v_1,\ldots, v_m\}}: & \widehat V & \longrightarrow & \widehat V\otimes \widehat V \\
 & \alpha_1 v_1+\ldots +\alpha_{m} v_m & \longmapsto  & \sum_{i=1}^{m}{\alpha_1^{q^{i-1}}v_{(i,i)}}+\ldots+\sum_{i=1}^{m}{\alpha_{m}^{q^{i-1}}v_{(i,m-1+i)}}.
 \end{array}
 \]
 Then $\Im(\widetilde\phi)$ is linearly equivalent to $\Im(\phi)$ in $\widehat V\otimes \widehat V$.
 \end{proposition}
\begin{proof}
Let $v=\sum_{i=1}^{m}{\alpha_i v_i}\in\widehat V$ be linear combination of $k$ vectors of rank 1, $1\le k\le m$. 

Let $\tau$ be the change of basis map of $\widehat V$ from the basis  $u_1,\ldots, u_m$ to the  Singer basis $v_1,\ldots, v_m$. 

Assume $k=1$, i.e. $v =  \lambda(\sum_{i=1}^{m}{a^{q^{i-1}}v_i})$, and set  $\lambda=\sum_{i=1}^{m}{l_iu_i}$,  $a=\sum_{i=1}^{m}{x_iu_i}$, with $l_{i},x_i\in\Fq$. Therefore, 
$v=\lambda\sum_{i=1}^{m}{x_iu_i}$ and
\[
\begin{array}{rcl}
\widetilde\phi(v) & =  & (\sum_{i=1}^{m}{\lambda^{q^{i-1}}v_i})\otimes (\sum_{i=1}^{m}{a^{q^{i-1}}v_i})\\[.03in]
  & = & (\sum_{i=1}^{m}{l_iu_i})^\tau\otimes (\sum_{i=1}^{m}{x_iu_i})^\tau\\[.03in]
  & = & [(\sum_{i=1}^{m}{l_iu_i}) \otimes (\sum_{i=1}^{m}{x_iu_i})]^{(\tau,\tau)}\\[.02in]
  & = & [\sum_{i=1}^{m}{l_ix_1u_{(i1)}}+\ldots+\sum_{i=1}^{m}{l_ix_mu_{(im)}}]^{(\tau,\tau)}\\[.03in]
  & = & \phi(v)^{(\tau,\tau)}.
    \end{array}
\]

Now assume $v  =  \lambda_1(\sum_{i=1}^{m}{a_1^{q^{i-1}}v_i})+\ldots+\lambda_k(\sum_{i=1}^{m}{a_k^{q^{i-1}}v_i})$, $k>1$. Set  $\lambda_j=\sum_{i=1}^{m}{l_{ij}u_i}$,  $a_j=\sum_{i=1}^{m}{x_{ij}u_i}$, with $l_{ij},x_{ij}\in\Fq$. Therefore, 
\[
v = \lambda_1(\sum_{i=1}^{m}{x_{i1}u_i})+\ldots+\lambda_k(\sum_{i=1}^{m}{x_{ik}u_i})=\sum_{i=1}^{m}{(\lambda_1x_{i1}+\ldots+\lambda_kx_{ik})u_i}
\]
giving $\phi(v)=\sum_{i=1}^{m}{(l_{i1}x_{11}+\ldots+l_{ik}x_{1k})u_{(i,1)}}+\ldots+\sum_{i=1}^{m}{(l_{i1}x_{m1}+\ldots+l_{ik}x_{mk})u_{(i,m)}}$.

On the other hand we have
\[
\begin{array}{rcl}
\widetilde\phi(v) &  = & (\sum_{i=1}^{m}{\lambda_1^{q^{i-1}}v_i})\otimes (\sum_{i=1}^{m}{a_1^{q^{i-1}}v_i})+\ldots+(\sum_{i=1}^{m}{\lambda_k^{q^{i-1}}v_i})\otimes (\sum_{i=1}^{m}{a_k^{q^{i-1}}v_i})\\[.03in]
& = & (\sum_{i=1}^{m}{l_{i1}u_i})^\tau\otimes (\sum_{i=1}^{m}{x_{i1}u_i})^\tau+\ldots+(\sum_{i=1}^{m}{l_{ik}u_i})^\tau\otimes (\sum_{i=1}^{m}{x_{ik}u_i})^\tau\\[.03in]
& = & [(\sum_{i=1}^{m}{l_{i1}u_i})\otimes (\sum_{i=1}^{m}{x_{i1}u_i})]^{(\tau,\tau)}+\ldots+[(\sum_{i=1}^{m}{l_{ik}u_i})\otimes (\sum_{i=1}^{m}{x_{ik}u_i})]^{(\tau,\tau)}\\[.03in]
& = & [\sum_{i=1}^{m}{l_{i1}x_{11}u_{(i1)}}+\ldots+\sum_{i=1}^{m}{l_{i1}x_{m1}u_{(im)}}]^{(\tau,\tau)}+\ldots\\[.03in]
& & +[\sum_{i=1}^{m}{l_{ik}x_{1k}u_{(i1)}}+\ldots +\sum_{i=1}^{m}{l_{ik}x_{mk}u_{(im)}}]^{(\tau,\tau)}\\[.03in]

& = & [\sum_{i=1}^{m}{(l_{i1}x_{11}+\ldots+l_{ik}x_{1k})}u_{(i1)}+\ldots +\sum_{i=1}^{m}{(l_{i1}x_{m1}+\ldots+l_{ik}x_{mk})}u_{(im)}]^{(\tau,\tau)}\\[.03in]
& = & \phi(v)^{(\tau,\tau)}.
\end{array}
\] 
\end{proof}

We call the map $\widetilde\phi$  the {\em field reduction of }  $\widehat V$ {\em over} $\Fq$ {\em with respect to the Singer  basis}  $v_1, \ldots, v_m$ and its image  the {\em  cyclic model for the field reduction of} $\widehat V$ over $\Fq$.  The projective space whose points are the $1$-dimensional  $\Fq-$subspaces generated by the elements of $\widetilde\phi(\widehat V)$ is  the {\em cyclic model for the field reduction} of $\PG(\widehat V)$ over $\Fq$.

Let $\widetilde\nu$  be the map defined by 
\[
\begin{array}{lccc}
\widetilde\nu=\nu_{\{v_1,\ldots, v_m\}}: & \widehat V\otimes \widehat V & \longrightarrow & M_{m,m}(\Fqm) \\
 & \sum_{i,j}{x_{ij}v_{(i,j)}}&   \longrightarrow  &   (x_{ij})_{i=1,\ldots,m}^{j=1,\ldots,m}.
 \end{array}
\]
Then, for any   $v=\alpha_1 v_1+\ldots+\alpha_{m} v_m\in  \widehat V$,  the matrix    $\widetilde \nu(\widetilde\phi(v))$ is  the Dickson matrix $D_{(\alpha_1,\ldots, \alpha_{m})}$. Since the cyclic model for the field reduction of $ \widehat V$ is obtained from the  field reduction   $\phi(\widehat V)$ by changing a basis in $\widehat V\otimes \widehat V$, we get that the rank of  $\widetilde\nu(\widetilde\phi(v))$ equals the rank of  $\nu(\phi(v))$, for any $v\in \widehat V$.


 In addition, the element  $k_v=\widetilde \phi(\<v\>)$ of the $m$-partition $\cK$ is
\[
k_{v}=\{\sum_{i=1}^{m}{(\lambda \alpha_1)^{q^{i-1}}v_{(i,i)}}+\ldots+\sum_{i=1}^{m}{(\lambda \alpha_{m})^{q^{i-1}}v_{(i,m-1+i)}}:\lambda\in\Fqm\}.
\]

\comment{
For any $\alpha\in\Fqm\setminus\{0\}$,  the linear transformation 
\begin{equation}\label{eq_50}
\alpha_1v_1+\alpha_2v_2+\alpha_3,\ldots \alpha_mv_m\mapsto \alpha_1v_1+\alpha \alpha_2v_2+\alpha^{1+q}\alpha_3v_3+\ldots+\alpha^{1+q+\ldots q^{m-1}} \alpha_mv_m
\end{equation}
  maps the set $\pi_a$, $a\neq 1$,  to $\pi_a$.
  }

 
In particular, $\bigcup_{v\in V\setminus\{0\}}{\widetilde\nu(k_v)}$ is the set of all rank 1 matrices in $\cD_m(\Fqm)$.

From  the arguments above,  we see that the set  $\cF_{m,q;I}$  can be considered,  via the isomorphism (\ref{eq_60}),  as the field reduction   of the set  $\cA_{m,q;I}$ with respect to the Singer basis $v_1,\ldots, v_m$.

As $[\pi_1]=[V]$, then the set $\cF_{\pi_1}=\widetilde\phi(\pi_1)$ defines the Segre variety  $\cS_{m,m}(\Fq)$
 of $\PG(V\otimes V)$ and  $\cF_{\pi_a}$ defines a Segre variety  projectively equivalent to $\cS_{m,m}(\Fq)$ under the element  of $\PGL(V\otimes V)$ corresponding to the linear transformation $\tau_\alpha$  with $N(\alpha)=a$. 
 
 \begin{remark}
 Note that, whenever $a\neq 1$, elements in $\cF_{\pi_a}$ have  rank bigger then 1 by Lemma \ref{prop_9}. This is explained by the fact that the linear transformation of $V\otimes V=V(m^2,q)$ corresponding to $\tau_\alpha$  is not in $\Aut_{\Fq}(V\otimes V)$.
 \end{remark}
 
Let  $W=\<v_1,v_m\>\subset \widehat V$.  Then $ \widetilde\phi(W)$  is a $2m$-dimensional vector subspace of $V\otimes V$. In $[\widetilde\phi(W)]$, the set  $[\widetilde\phi(J_1)]$ is  the Bruck norm-surface 
\[
\cN=\cN_{(-1)^m}=\{[\widetilde\phi(xv_1+yv_m)]:x,y \in \Fqm,\  N(y/x)=(-1)^m\}
\]
introduced in \cite{b1} and widely investigated in \cite{b2,b3} and recently in \cite{cz,lsz}. For any $x\in\Fqm\setminus\{0\}$ set $J_x=\{\lambda x v_1-\lambda x^{q^{m-1}}v_m:\lambda \in \Fqm\}$. 
  Then $ [\widetilde\phi(J_x)]\subset \cN$ and the set $\{[\widetilde\phi(J_x)]:x \in \Fqm\}$ is    a so-called {\em hyper-regulus} of  $\PG(\widetilde W)$ \cite{ost}. It turns out, that under the linear transformation $\tau_\alpha$ with $N(\alpha)=a$, also $J_a$ defines a hyper-regulus of  $[\widetilde\phi(W)]$.

The following result, which  summarizes all above arguments,  gives a geometric description of the MRD codes $\cF_{m,q;I}$.

\begin{theorem} \label{th_5}
Let $q>2$ be a prime power and $m>2$ a positive integer. Let  $I$  be any nonempty subset of $\Fq\setminus\{0,1\}$ with $k=|I|$. The projective image of the MRD code $\cF_{m,q;I}$  in $\PG(m^2-1,q)$ is a subset of a  Desarguesian spread  which is union of two spread elements, $k$ mutually disjoint Segre varieties $\cS_{m,m}(\Fq)$ and $q-1-k$ mutually disjoint hypereguli all contained in the $(2m-1)$-dimensional projective subspace generated by the two spread elements.
\end{theorem}

%
%
%
%

\section{The Cossidente-Marino-Pavese non-linear MRD code}\label{sec_5}

Recently, Cossidente, Marino and Pavese constructed  non-linear $(3,3,q;1)$-MRD codesin a totally geometric setting \cite[Theorem 3.6]{cmp}. 

In $\PG(2,q^3)$,  $q \ge 3$,  let  $\cC$ be the set of points whose coordinates satisfy the equation $X_1 X_2^q - X_3^{q+1}=0$, that
is a  {\em $C_F^1$-set} of $\PG(2,q^3)$ as introduced and studied in  \cite{dd}.  The set $\cC$ is the projective image of a subset  of $V(3,q^3)$ which is the union of $A_1$, $A_2'=\{(0,x,0):x \in\Fqthree\setminus\{0\}\}$  and the $q-1$ sets $ \gamma_a =\{(\lambda , \lambda x^{q+1}, \lambda x^{q}): \lambda,x\in\Fqthree\setminus\{0\}, N(x)=a\}$,  with $a$ a nonzero element of $\Fq$.  

For any nonzero $a\in\Fq$, let $\alpha\in\Fqthree$ with $N(\alpha)=a$ and  set $Z_a=\{(\lambda x,-\lambda \alpha x^q,0):\lambda,x\in\Fqthree\setminus\{0\}\}$. Let $I$ be any non-empty  subset  of $\Fq\setminus\{0,1\}$ and  put  
\[
\cA'(q;I)=\bigcup_{a\in I}{\gamma_a}\bigcup_{b\in \Fq \setminus (I\cup\{0\})}{Z_b}\cup A_1\cup A_2'\cup\{\0\}.
\]
Up to an endomorphism of $V\otimes V$ viewed as the vector space $V(9,q)$,  the image of set $\cA'(q;I)$ under $\nu\circ\phi$ is  a non-linear $(3,3,q;1)$-MRD code \cite[Proposition 3.8]{cmp}. 

\begin{lemma}\label{lem_8}
Let $\theta$ be the semilinear transformation of $V(3,q^3)$ defined by  
\[
\begin{array}{rccc}
\theta: & v_1 & \mapsto & v_3 \\
        & v_2 & \mapsto & v_1 \\
        & v_3 & \mapsto & v_2  
        \end{array} 
\]
with associated automorphism $x\mapsto x^{q^2}$.  Then  $\theta$ maps $\gamma_a$ into $ \pi_{a^{-1}}$ and $Z_a$ into $J_{a^{-1}}$, for any nonzero element $a$ of $\Fq$.  
\end{lemma}
\begin{proof}
Every element $x\in\Fqthree$ with $N(x)=a$ can be written as $x=\alpha t^{q-1}$ for some  $t\in\Fqthree$ and $\alpha$ a fixed element in  $\Fqthree$ such that $N(\alpha)=a$. By straightforward 
calculations, we can write $\gamma_a=\{(\lambda x,\lambda \alpha^{q+1} x^{q},\lambda \alpha^q x^{q^2}):\lambda,x\in\Fqthree\}$. 
Then, we get $\theta(\gamma_a)= \{(\lambda x,\lambda(\alpha^{-1})^{q^2} x^q, \lambda (\alpha^{-1})^{(q^2+1)} x^{q^2} ):\lambda,x\in\Fqthree\}=\pi_{a^{-1}}$ as $N(\alpha^{-q^2})=N(\alpha^{-1})=a^{-1}$. 

The last part of the statement follows from straightforward calculations.
\end{proof}

\begin{corollary}
Let $I$ be any  non-empty subset $I$ of $\Fq\setminus\{0,1\}$ and  put $I^{-1}=\{a^{-1}:a\in I\}$. Then, up to the endomorphism $\theta$ of $V(3,q^3)$ and the changing of basis in $V(3,q^3)\otimes V(3,q^3)$ from $u_{(i,j)}$ to $v_{(i,j)}$, the Cossidente-Marino-Pavese family of non-linear MRD codes is the set $\cF_{3,q,I^{-1}}$.   
\end{corollary}

Let $L$ be any line of $\PG(2,q^3)$ disjoint from a subgeometry $\PG(2,q)$.  The set of points of $L$ that lie on some proper subspace spanned by points of $\PG(2,q)$ is called the {\em exterior splash} of $\PG(2,q)$ on  $L$ \cite{lz}. 

\begin{proposition}\cite{cz}
The exterior splash of the subgeometry $[\pi_a]$ on the line $[W]$ is the set 
$[J_b]$ with 
$b=a^{m-1}$. 
\end{proposition}
\begin{proof}
First we note that $[W]$ is disjoint from  $ [\pi_1]$.
The $\Fqm$-span of some hyperplane in  the cyclic model of $V$ is a hyperplane of $\widehat V$ with  equation $\sum_{i=1}^{m}{\alpha^{q^{i-1}}X_{i}}=0$, for some nonzero$\alpha\in\Fqm$. As the Singer cycle $\sigma$ acts  on the hyperplanes of $V$ by mapping the hyperplane  with equation  $\sum_{i=1}^{m}{\alpha^{q^{i-1}}X^{q^{i-1}}}=0$ to the hyperplane with equation $\sum_{i=1}^{m}{(\mu\alpha)^{q^{i-1}}X^{q^{i-1}}}=0$, then $\sigma$ maps the hyperplane  of $\widehat V$ with equation $\sum_{i=1}^{m}{\alpha^{q^{i-1}}X_{i}}=0$ into the hyperplane with equation $\sum_{i=1}^{m}{(\mu\alpha)^{q^{i-1}}X_i}=0$. Note that  $\sigma$ fixes  $W$. 

The hyperplane $\sum_{i=1}^{m}{X_i}=0$ of $\widehat V$ meets  $W$ in the $\Fqm$-subspace spanned by $v_1-v_m$. By looking at  the action of the Singer cyclic group $S=\<\sigma\>$ on $W$,  we see that the exterior splash of $ [\pi_1]$ on $[W]$ is the set  $[J_1]$.
By using he map $\tau_\alpha$ defined above with $N(\alpha)=a$, we get the result.
\end{proof}

\begin{remark}
Let $U$ be the $\Fqm$-span of $v_1$ and $v_2$ in $\widehat V$. It is evident that the semilinear transformation $\theta$ maps the exterior splash of $[\gamma_a]$ on $[U]$ into the exterior splash of $[\pi_{a^{-1}}]$ on $[W]$.  

 The exterior splash of $[\gamma_a]$ on $[U]$ is 
\[
[\gamma_a] =\{[(1, x,0)]: x\in\Fqthree, N (x)=-a^2\}.
\]

In $\cite{cmp}$, the splash of $[\gamma_a]$ was erroneusly given as the set $[Z_a]$. Note that,  $[Z_a]$   never coincides with $[\gamma_a]$, unless $a=1$.
\end{remark}


\begin{thebibliography}{10}
%
\bibitem{alr} D. Augot, P. Loidreau, G. Robert, Rank metric and Gabidulin codes in characteristic zero, {\em Proceedings  ISIT 2013}, 509--513.

%
 \bibitem{bl} L. Bader, G.  Lunardon,
Some remarks on the Spin Module Representation of $\Sp_6(2^e)$,
{\em Discrete Math.} {\bf 339} (2016),  1265--1273.
%
\bibitem{b1} R.H. Bruck, Construction problems of finite projective planes, in: Proc. Conf. on Combinatorics, University of North Carolina at Chapell Hill, April 10--14, 1967, University of North Carolina Press, Chapel Hill, 1969, pp. 426--514.
%
\bibitem{b2} R.H. Bruck,  Circle geometry in higher dimensions. II. {\em Geometriae Dedicata} {\bf 2} (1973), 133--188.
%
\bibitem{b3} R.H. Bruck,  The automorphism group of a circle geometry, {\em J. Combin. Theory Ser. A} {\bf 32} (1982), 256--263.
%
\bibitem{car} L. Carliz, A Note on the Betti-Mathieu group, {\em Portugaliae mathematica} {\bf 22 (3)}  (1963), 121--125.
%
%
\bibitem{coop} B.N. Cooperstein, External flats to varieties in $\PG(M_{n,n}\GF(q))$, {\em Linear Algebra Appl.}  {\bf 267}  (1997), 175--186.
%
\bibitem{cmp} A. Cossidente, G. Marino, F. Pavese, Non-linear maximum rank distance codes, {\em Des. Codes Cryptogr.},  {\bf 79 (3)} (2016), 597-Ð609.
%
\bibitem{ckww} J. de la Cruz, M. Kiermaier, A. Wassermann, W. Willems, Algebraic structures of MRD Codes,  {\em Adv. Math. Commun.} {\bf 10} (2016),  499--510.
%
%
%
\bibitem{cz} B. Csajb\'ok , C. Zanella,   On scattered linear sets of pseudoregulus type in $\PG(1,q^t)$, {\em Finite Fields Appl.} {\bf 41} (2016), 34--54. 
%
\bibitem{ds}  Ph. Delsarte,  Bilinear forms over a finite field, with applications to coding theory, {\em  J. Combin. Theory Ser. A} {\bf 25} (1978),  226--241.
%
\bibitem{d} P. Dembowski, Finite Geometries. {\em Springer 1968}.
%
\bibitem{dd}  G. Donati, N. Durante,  Scattered linear sets generated by collineations between pencils of lines, {\em J. Algebraic Combin.} {\bf 40} (2014), 1121--1134.
%

\bibitem{dye} R.H. Dye,  Spreads and classes of maximum subgroups of $\GL_n(q)$, $\SL_n(q)$, $\PGL_n(q)$ and $\PSL_n(q)$, {\em Ann. Mat. Pura Appl. (4)} {\bf 158} (1991), 33--50.
%
\bibitem{fkmp} G. Faina, G. Kiss, S.  Marcugini, F. Pambianco, 
The cyclic model for $\PG(n,q)$ and a construction of arcs, {\em 
European J. Combin.} {\bf 23} (2002),  31--35.
%
\bibitem{gab}  E.M. Gabidulin, Theory of codes with maximum rank distance, {\em Problemy Peredachi Informatsii} {\bf 21} (1985), 3--16.
%
\bibitem{gpt2}  E.M. Gabidulin, A.V. Paramonov, O.V. Tretjakov, 
Ideals over a noncommutative ring and their application in cryptology, Advances in cryptology, EUROCRYPT '91, {\em  Lecture Notes in Comput. Sci.} {\bf 547} (1991), 482--489. 
%
\bibitem{gpt}  E.M. Gabidulin, A.V. Paramonov, O.V. Tretjakov, Rank errors and rank erasures correction, Proceedings of the 4th International Colloquium on Coding Theory, Dilijan, Armenia, Yerevan, 1992, pp. 11--19.
%
%
\bibitem{h1} J.W.P. Hirschfeld, Projective Geometries Over Finite Fields, 2nd edn, Clarendon Press, Oxford, 1998.
%
\bibitem{ht} J.W.P. Hirschfeld, J.A. Thas, {\em General Galois Geometries},
Oxford University Press, New York, 1991.
%
%
\bibitem{kg} A. Kshevetskiy and E. M. Gabidulin, The new construction of rank codes, Proc. IEEE Int. Symp. on Information Theory,pp. 2105--2108, Sept. 2005.
%
\bibitem{kk} R. K\"otter, F. Kschischang,
Coding for errors and erasures in random network coding, {\em IEEE Trans. Inform. Theory} {\bf 54} (2008), 3579--3591.
%
\bibitem{lsz} M. Lavrauw, J. Sheekey, C. Zanella, On embeddings of minimum dimension of $\PG(n,q)\otimes\PG(n,q)$, {\em Des. Codes Cryptogr.} {\bf 74} (2015), 427--440.
%
\bibitem{lvdv} M. Lavrauw, G. Van de Voorde, Scattered linear sets and pseudoreguli, {\em Electron. J. Combin.} {\bf 20} (2013),  Paper 15, 14 pp. 
%
%
\bibitem{lz} M. Lavrauw, C.  Zanella,  Subgeometries and linear sets on a projective line, {\em  Finite Fields Appl.} {\bf 34} (2015), 95--106.
%
\bibitem{ln} R. Lidl, H. Niederreiter,  Finite fields.  Encyclopedia of Mathematics and its Applications, 20. Cambridge University Press, Cambridge, 1997.
 %
 \bibitem{lun} G. Lunardon, MRD-codes and linear sets, {\em  J. Combin. Theory Ser. A} {\bf  149} (2017), 1--20.
 %
 \bibitem{ltz} G. Lunardon, R. Trombetti, Y. Zhou, Generalized twisted Gabidulin codes, preprint, http://arxiv.org/pdf/1507.07855v2.pdf.
 %
 \bibitem{mpt} G. Marino, O. Polverino, R. Trombetti, Maximum scattered linear sets of pseudoregulus type and the Segre variety ${\cal S}_{n,n}$,
{\em  J. Algebraic Combin.} {\bf 39 }(2014), no. 4, 807--831.
 %
\bibitem{ost} T.G. Ostrom, Hyper-reguli, {\em J. Geom.} {\bf 48} (1993), 157--166.
%
\bibitem{pol} O. Polverino,   Linear sets in finite projective spaces, {\em Discrete Math.} {\bf 310} (2010),  3096--3107.
%
\bibitem{rav} A. Ravagnani, Rank-metric codes and their duality theory, {\em  Des. Codes Cryptogr.} {\bf 80} (2016), 197--216. 
 %
 \bibitem{roman}  S. Roman,  Field theory. Graduate Texts in Mathematics, 158. Springer, New York, 2006. 
 %
 \bibitem{segre} B. Segre, Teoria di Galois, fibrazioni proiettive e geometrie non desarguesiane, {\em Ann. Mat. Pura Appl.} {\bf 64} (1964), 1--76.
 %
  \bibitem{she} J. Sheekey, A new family of linear maximum rank distance codes, {\em  Adv. Math. Commun.} {\bf 10} (2016),  475--488. 
 %
\bibitem{sk}D. Silva, F.R. Kschischang, Universal Secure Network Coding
via Rank-Metric Codes,  {\em IEEE Trans. Inform. Theory} {\bf 57} (2011),  1124--1135.
 %
 \bibitem{skk} D. Silva, F.R.  Kschischang, R. K\"otter,  A rank-metric approach to error control in random network coding, {\em IEEE Trans. Inform. Theory} {\bf 54} (2008),  3951--3967.
 %
  
 %
%
\bibitem{tsc} V. Tarokh, N. Seshadri, A.R. Calderbank,  Space-time codes for high data rate wireless communication: performance criterion and code construction, {\em IEEE Trans. Inform. Theory 44 (1998)},  744--765.
%



\bibitem{wl} B. Wu, Z. Liu, Linearized polynomials over finite fields revisited, {\em Finite Fields Appl.} {\bf 22} (2013), 79--100. 
%
\end{thebibliography}
\end{document}